\documentclass[conference]{IEEEtran}
\IEEEoverridecommandlockouts

\usepackage{cite}
\usepackage{amsmath,amssymb,amsfonts}
\usepackage{amsthm}
\usepackage{booktabs}
\usepackage{graphicx}
\usepackage{array}
\usepackage{url}
\usepackage{xcolor}

\usepackage[utf8]{inputenc}
\usepackage[T1]{fontenc}
\usepackage{microtype}

\usepackage{algorithm}
\usepackage{algpseudocode}
\usepackage{multirow}

\newtheorem{definition}{Definition}
\newtheorem{proposition}{Proposition}
\newtheorem{theorem}{Theorem}

\newcommand{\Hilb}{\mathcal{H}}
\newcommand{\Tr}{\operatorname{Tr}}
\newcommand{\QLTM}{QL-TM}
\newcommand{\rhoin}{\rho}

\begin{document}

\title{Quantum-Logic Tsetlin Machines:\\Interpretable Quantum Machine Learning with Commuting Projector Clauses}

\author{\IEEEauthorblockN{Krishna Bhatia}
\IEEEauthorblockA{QuantumAI Lab, Fractal Analytics, Mumbai, India}}

\maketitle

\begin{abstract}
Tsetlin Machines learn propositional clauses over Boolean literals using teams of finite-state automata. This paper introduces a quantum-logic analogue, the Quantum-Logic Tsetlin Machine (\QLTM), in which literals are quantum propositions represented by projection operators and clauses are conjunctions inside commuting measurement contexts. The activation of a clause on a quantum state is the Born probability of the corresponding joint projector. This gives a hybrid, interpretable quantum machine-learning model: the automata remain classical, but the learned rules are quantum propositions over states rather than post-measurement bitstrings. We prove that the model reduces to an ordinary Boolean Tsetlin Machine in diagonal contexts, and that Pauli projector clauses recover stabilizer and syndrome semantics. Experiments on Bell states, phase-flip syndrome states, randomized stabilizer-context tasks, mixed literal pools, and finite-shot noisy simulations show that QL projector clauses attain 1.000 accuracy in the exact/noiseless controlled settings, while diagonal or wrong-context baselines lose the relevant phase or syndrome information and remain far below the true context. In 16-class random stabilizer tasks, the true QL context attains 1.000 accuracy for five- and six-qubit systems, while wrong QL contexts remain near the 0.0625 chance level. Mixed-pool and context-budget ablations further show that the learner recovers all true context generators in the presence of distractor projectors, and that accuracy degrades predictably when true quantum propositions are removed. The contribution is not a claim of quantum advantage, but a controlled bridge between Tsetlin clause learning, quantum logic, and interpretable quantum machine learning.
\end{abstract}


\section{Introduction}

Tsetlin Machines (TMs) are logic-based machine-learning models in which classification is expressed through propositional clauses learned by Tsetlin automata~\cite{Granmo2018TM}. Their appeal is that the learned model is not a dense numerical representation, but a set of human-readable rules such as conjunctions over Boolean literals. This makes TMs attractive in settings where interpretability, discrete learning, and hardware-near implementation matter.

Quantum machine learning (QML), in contrast, usually represents information through quantum states, measurements, kernels, or variational circuits~\cite{Biamonte2017QML,Schuld2019Feature,Havlicek2019Kernel,Cerezo2021VQA}. These models can exploit quantum state spaces, but their learned objects are often less directly logical. This raises a natural question: if a classical TM learns Boolean propositions, what should a TM-like model learn when the input is a quantum state?

A direct ``fully quantum'' Tsetlin Machine is difficult to define. Classical TM learning uses irreversible reward and penalty updates, while closed quantum dynamics is unitary unless measurement, feedback, or coupling to an environment is introduced. We therefore take a conservative route: we quantize the \emph{logic} of the TM rather than the automata memory. In the quantum-logic view of Birkhoff and von Neumann, propositions about a quantum system correspond to closed subspaces of Hilbert space, equivalently projectors~\cite{Birkhoff1936QuantumLogic,WilceQuantumLogic}. This suggests replacing Boolean literals by projector literals.

The resulting model, the Quantum-Logic Tsetlin Machine (\QLTM), is a hybrid QML model. Its inputs are quantum states, its literals are projective measurement propositions, and its clauses are commuting conjunctions of such propositions. The Tsetlin automata remain classical finite-state learners that include or exclude literals. Fig.~\ref{fig:pipeline} summarizes the model. The model is not a quantum Turing machine~\cite{Deutsch1985QuantumComputer}; the acronym \QLTM{} is used here only for ``Quantum-Logic Tsetlin Machine.''

\begin{figure*}[t]
    \centering
    \includegraphics[width=\textwidth]{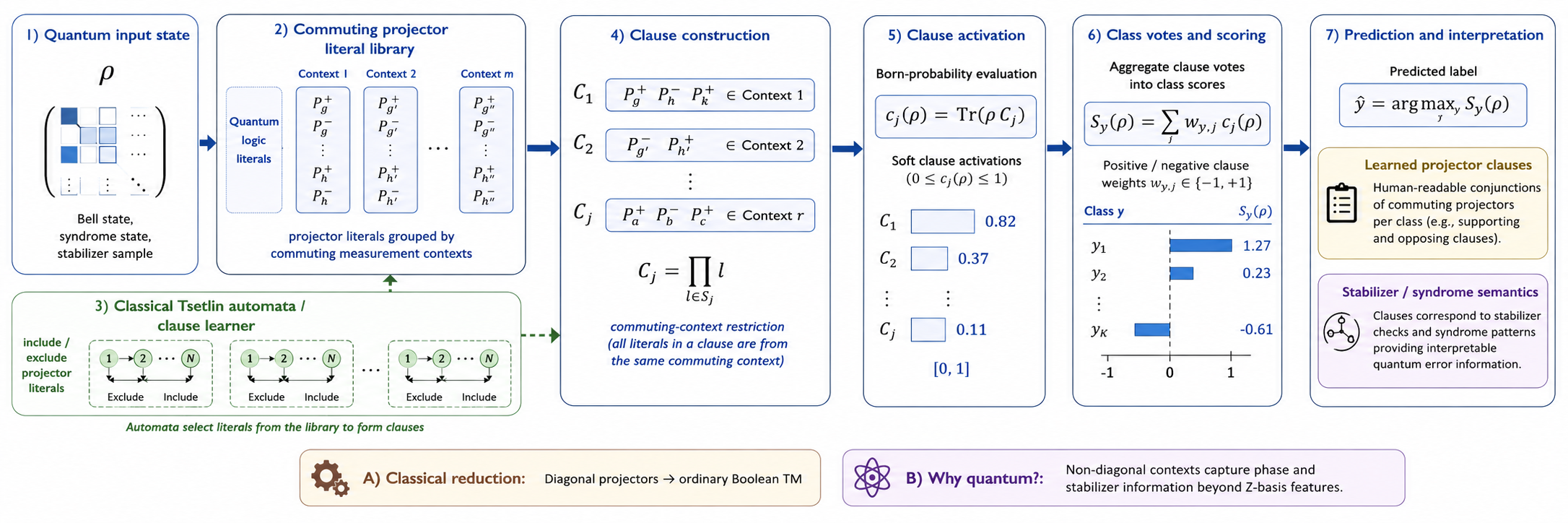}
    \caption{QL-TM pipeline. A quantum state is evaluated through a library of commuting projector literals. Automata select literals into clauses, each clause produces a Born-probability activation, and class scores are formed from clause votes.}
    \label{fig:pipeline}
\end{figure*}

\textbf{Contributions.} This paper makes four contributions:
\begin{enumerate}
    \item It defines TM clauses over quantum projectors, with Born-probability activations and an explicit commuting-context restriction.
    \item It proves a classical-reduction property: when all projectors are diagonal in a common computational basis, QL-TM clauses reduce exactly to Boolean TM clauses.
    \item It connects Pauli-projector clauses to stabilizer and syndrome rules, giving a physical interpretation of learned clauses.
    \item It provides controlled experiments showing contextual expressivity, randomized stabilizer generalization, recovery of true projectors among distractors, predictable degradation under context-budget removal, and graceful degradation under finite-shot noise.
\end{enumerate}

The goal is deliberately modest but foundational: define a mathematically controlled quantum-logic generalization of Tsetlin clause learning and test whether the resulting clauses recover physically meaningful quantum propositions.

\section{Background and Related Work}

\subsection{Tsetlin Machines}

A classical TM learns a set of clauses over Boolean literals $x_i$ and $\neg x_i$~\cite{Granmo2018TM}. Each clause is a conjunction of included literals. Teams of Tsetlin automata decide whether each literal is included or excluded, and Type I/Type II feedback drives clauses toward discriminative propositional patterns. The learned representation is therefore symbolic: a classifier can be inspected as a weighted collection of logical rules. Classical TM software ecosystems such as TMU provide scalable implementations and variants~\cite{TMU2025}; our experimental implementation is independent and minimal because the literal semantics here are quantum projectors rather than ordinary Boolean variables.

\subsection{Quantum Logic and Projectors}

In quantum mechanics, a proposition such as ``observable $A$ has value in set $S$'' is represented by a projector onto the corresponding subspace. The set of all such propositions is not generally Boolean: projectors need not commute, and the lattice of quantum propositions is generally non-distributive~\cite{Birkhoff1936QuantumLogic,WilceQuantumLogic}. This is exactly where a QL-TM must differ from a classical TM. An arbitrary conjunction of non-commuting quantum propositions is not a simple Boolean AND. We therefore restrict each clause to a commuting context, where products of projectors are again projectors and the usual conjunction semantics are well defined.

\subsection{Pauli Projectors and Stabilizers}

For a Pauli observable $g$ with eigenvalues $\pm 1$, the positive and negative projector literals are
\begin{equation}
    P_g^{\pm}=\frac{I \pm g}{2}.
\end{equation}
If $g_1,\ldots,g_k$ commute, the product
\begin{equation}
    C_s = \prod_{i=1}^{k} P_{g_i}^{s_i}, \quad s_i\in\{+,-\},
\end{equation}
projects onto a joint eigenspace. This is the same semantic structure used by stabilizer codes, syndrome measurements, and quantum error correction~\cite{NielsenChuang2010,Gottesman1997Stabilizer,Aaronson2004CHP,Calderbank1996GoodQuantumCodes,Steane1996ErrorCorrecting}. A learned QL-TM clause can therefore be read as a stabilizer-like proposition: it identifies a syndrome pattern, not merely a numerical feature vector.

\subsection{Relation to Existing QML Models and Research Gap}

Variational classifiers and quantum kernels learn continuous circuit/readout parameters or implicit similarities rather than explicit logical propositions~\cite{Biamonte2017QML,Schuld2019Feature,Havlicek2019Kernel,Cerezo2021VQA}. Our question is different: can TM clause semantics be lifted from Boolean events to quantum propositions while preserving an interpretable include/exclude learner? The baselines therefore separate measurement-context information, projector-feature expressivity, and clause recovery. We do not claim predictive superiority over VQCs or quantum kernels; end-to-end comparisons on less structured state-classification tasks are future work. Prior quantum-adjacent TM work also does not, to our knowledge, define TM literals themselves as Hilbert-space projectors. In \QLTM, the automata remain classical while the learned clauses become explicit commuting quantum propositions.

\section{Quantum-Logic Tsetlin Machine}

Let $\Hilb$ be a finite-dimensional Hilbert space and let $\rhoin$ be a density operator on $\Hilb$. A quantum proposition is represented by a projector $P=P^2=P^\dagger$. Its complement is $P^\perp=I-P$.

\begin{definition}[Projector literal]
Given a base projector $P_i$, a QL-TM literal is either $P_i$, its complement $I-P_i$, or the excluded action $\varnothing$.
\end{definition}

\begin{definition}[Commuting context]
A context is a set of projectors $\mathcal{C}=\{P_1,\ldots,P_m\}$ such that $[P_i,P_j]=0$ for all $i,j$.
\end{definition}

\begin{definition}[Projector clause]
A projector clause is a product
\begin{equation}
    C_j = \prod_{i\in S_j} L_{ji}, \qquad L_{ji}\in \{P_i,I-P_i\},
    \label{eq:clause}
\end{equation}
where all selected literals belong to a common commuting context. Its activation on input state $\rhoin$ is
\begin{equation}
    c_j(\rhoin)=\Tr(\rhoin C_j) \in [0,1].
    \label{eq:activation}
\end{equation}
\end{definition}

Thus the binary output of a classical Boolean clause becomes a Born probability. In an exact simulation, $c_j(\rhoin)$ can be used directly as a soft clause vote. In an experimental implementation, $c_j(\rhoin)$ can be estimated by repeated measurements in the clause context.

For multiclass classification, class $y$ has a signed set of clauses, and the score is
\begin{equation}
    S_y(\rhoin) = \sum_{j} w_{yj} c_j(\rhoin), \qquad w_{yj}\in\{-1,+1\},
    \label{eq:score}
\end{equation}
with prediction $\hat y=\arg\max_y S_y(\rhoin)$. In the prototype used here, each class is represented by one positive clause, which is sufficient for the stabilizer-style tasks considered below.

\subsection{Automata Learning Interface}

For each class $y$, clause $j$, and signed projector literal $L_i$, the learner stores an automaton $a_{yji}\in\{1,\ldots,2N\}$ and includes $L_i$ iff $a_{yji}>N$. Signed literals $P_g^+$ and $P_g^-$ are represented separately, with conflict resolution preventing both signs of the same observable from remaining in one clause. Each literal has soft truth value
\begin{equation}
    p_i(\rho)=\Tr(\rho L_i)\in[0,1],
\end{equation}
and the included set $S_{yj}=\{i:a_{yji}>N\}$ defines $C_{yj}=\prod_{i\in S_{yj}}L_i$.

The formal clause activation is the joint Born probability $c_{yj}(\rho)=\Tr(\rho C_{yj})$. The minimal learner uses
\begin{equation}
    \tilde c_{yj}(\rho)=\prod_{i\in S_{yj}}\Tr(\rho L_i)
\end{equation}
only as a lightweight training surrogate. In general $\Tr(\rho\prod_iL_i)\neq\prod_i\Tr(\rho L_i)$, even for commuting projectors, so no independence assumption is made. For the exact stabilizer-eigenstate tasks central to our representational claims, relevant signed stabilizer projectors have probabilities in $\{0,1\}$; hence the surrogate and joint activation induce the same binary clause ordering. This equivalence is not guaranteed for finite-shot, noisy, or general mixed states, where direct joint-context estimation is a natural extension.

Algorithm~\ref{alg:training} states the update used in the experiments. Target-class clauses receive Type-I-like feedback: high-probability literals are rewarded by moving their automata toward the include region. Non-target clauses that fire receive Type-II-like feedback: the learner searches for a low-probability literal that suppresses the non-target clause and rewards inclusion of that suppressing literal. Same-observable sign conflicts are resolved after each update.

\begin{algorithm}[t]
\caption{Minimal QL-TM training rule}
\label{alg:training}
\footnotesize
\begin{algorithmic}[1]
\Require $\mathcal D=\{(\rho,y)\}$; literals $\{L_i\}$ by context; states $a_{yji}$; threshold $\tau$; epochs $E$
\Ensure Learned literal sets $S_{yj}$

\State Initialize automata near the include/exclude boundary.

\For{$e=1,\ldots,E$}
  \ForAll{$(\rho,y)\in\mathcal D$}
    \State Estimate $p_i=\Tr(\rho L_i)$ in the selected contexts.

    \ForAll{classes $y'$ and clauses $j$}
      \State $S_{y'j}\gets\{i:a_{y'ji}>N\}$, restricted to one context.
      \State $\tilde c_{y'j}\gets\prod_{i\in S_{y'j}}p_i$.
    \EndFor

    \ForAll{target clauses $j$ and candidate literals $L_i$}
      \State If $p_i\ge\tau$, move $a_{yji}$ toward Include; optionally penalize low-$p_i$ included literals.
    \EndFor

    \ForAll{$y'\neq y$ and firing clauses $j$ with $\tilde c_{y'j}>\tau$}
      \State Move a compatible low-probability suppressor $a_{y'jq}$ toward Include.
    \EndFor

    \State Resolve $P_g^+/P_g^-$ conflicts by retaining the stronger inclusion state.
  \EndFor
\EndFor

\State Predict with $c_{yj}(\rho)=\Tr(\rho C_{yj})$,
$S_y(\rho)=\sum_jw_{yj}c_{yj}(\rho)$, and
$\hat y=\arg\max_y S_y(\rho)$.

\end{algorithmic}
\end{algorithm}

\paragraph{Relation to canonical TM feedback.}
The prototype is intentionally TM-style rather than a complete reproduction of production score-gated TM feedback. Target updates depend locally on projector probabilities, non-target suppression is triggered by $\tilde c_{yj}$, and aggregate class scores are used for prediction. We therefore use ``Type-I-like'' and ``Type-II-like'' deliberately. The semantic change is precise: Boolean satisfaction $x_i\in\{0,1\}$ is replaced by projector satisfaction $p_i(\rho)=\Tr(\rho L_i)\in[0,1]$, while discrete include/exclude automata are retained.

The experiments use \emph{ProjectorClauseMiner}, a transparent margin-based clause learner, and \emph{TinyProjectorTsetlin}, the minimal automata prototype above. TinyProjectorTsetlin is not a drop-in replacement for production TM libraries; it tests whether TM-style automata can recover meaningful quantum-logic clauses. Quantum mechanics supplies states and measurement propositions, while classical automata perform clause search; commuting-context clauses can in principle be estimated from compatible measurement counts.

\begin{figure*}[t]
    \centering
    \includegraphics[width=\textwidth]{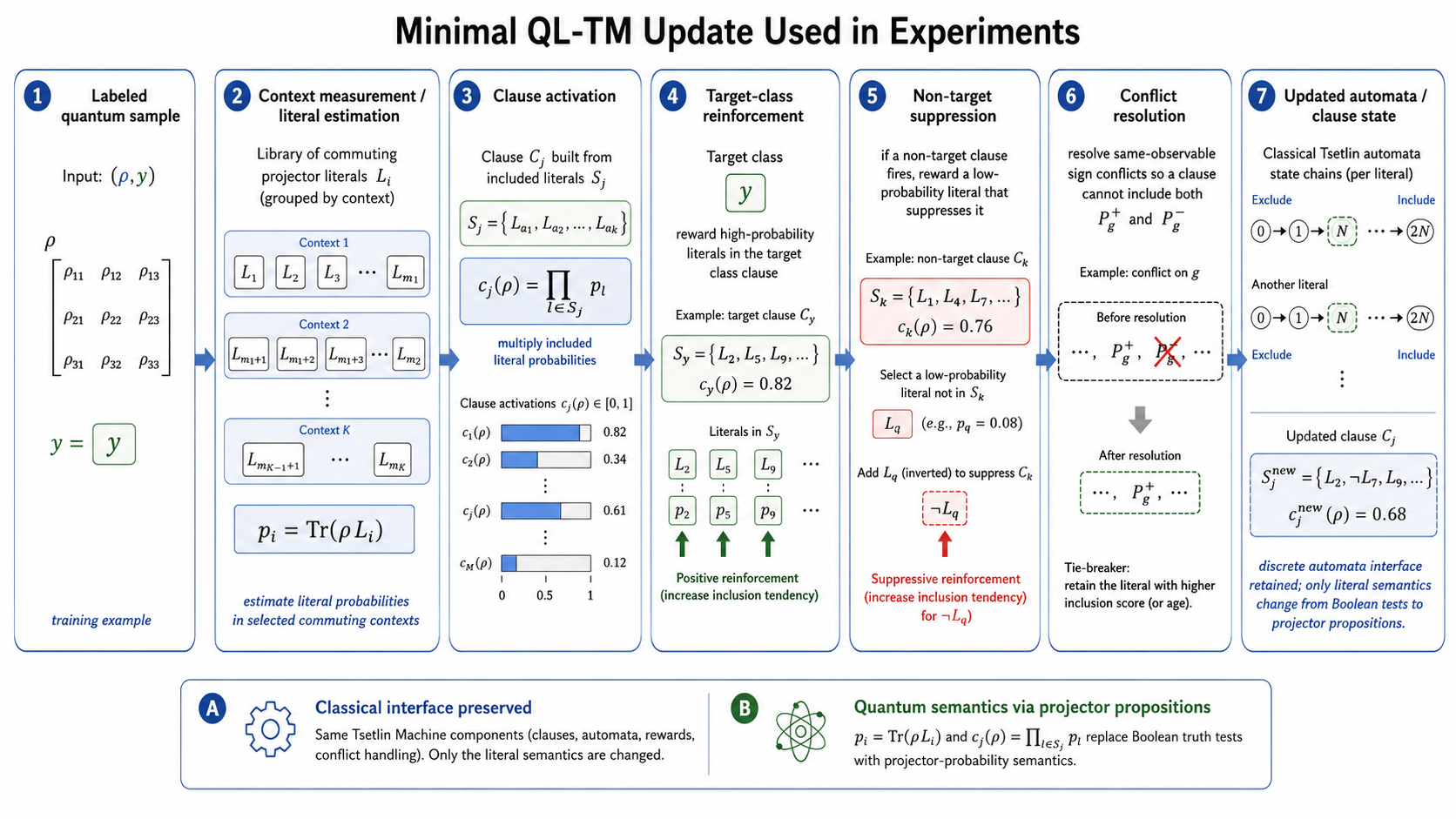}
    \caption{Minimal QL-TM update used in the experiments. Classical finite-state include/exclude automata are retained while Boolean literal truth values are replaced by projector Born probabilities within commuting contexts.}
    \label{fig:update}
\end{figure*}

\subsection{Classical Reduction}

\begin{theorem}[Diagonal classical reduction]
Suppose all base projectors $P_i$ are diagonal in the computational basis and correspond to events $x_i=1$ for Boolean variables $x_i\in\{0,1\}$. For any computational-basis state $|x\rangle$, every QL-TM clause activation equals the value of the corresponding Boolean conjunction:
\begin{equation}
    \Tr(|x\rangle\langle x| C_j) = \bigwedge_{i\in S_j} \ell_i(x),
\end{equation}
where $\ell_i(x)$ is either $x_i$ or $\neg x_i$.
\end{theorem}

\begin{proof}
Since all projectors are diagonal, they commute and act as indicator functions on computational-basis states. A product of selected projectors is therefore the product of the corresponding indicators, which is exactly Boolean AND. Complements $I-P_i$ act as negated literals. Hence the QL-TM score reduces to a classical TM score over Boolean literals.
\end{proof}

This theorem is important for positioning: QL-TM does not discard the classical TM. It extends its clause semantics from Boolean events to commuting quantum propositions.

\subsection{Pauli-Stabilizer Clause Semantics}

\begin{proposition}[Stabilizer clause]
Let $g_1,\ldots,g_k$ be commuting Pauli operators and let $s_i\in\{+1,-1\}$. The clause
\begin{equation}
    C_s = \prod_{i=1}^k \frac{I+s_i g_i}{2}
\end{equation}
projects onto the joint eigenspace with eigenvalue $s_i$ for each $g_i$.
\end{proposition}

\begin{proof}
Each factor is the spectral projector onto the $s_i$ eigenspace of $g_i$. Because the $g_i$ commute, their spectral projectors commute. The product is therefore the projector onto the intersection of these eigenspaces.
\end{proof}

\section{Experimental Design}

\subsection{Research Questions}

The experiments answer five questions. First, do QL projector contexts give information unavailable to diagonal classical features? Second, can a QL-TM recover known physical clauses such as Bell-state and syndrome propositions? Third, does the result generalize to randomized stabilizer contexts rather than hand-picked examples? Fourth, if the learner is given a mixed pool containing true, wrong, diagonal, and random distractor literals, does it select the true quantum-context generators? Fifth, if some true generators are removed entirely, does performance degrade according to the remaining stabilizer information?

\subsection{Tasks}

\textbf{Bell task.} The four classes are $|\Phi^+\rangle$, $|\Phi^-\rangle$, $|\Psi^+\rangle$, and $|\Psi^-\rangle$. The informative QL context contains $ZZ$ and $XX$ projectors. A diagonal baseline only sees $Z$-basis projectors and cannot distinguish the phase labels.

\textbf{Phase-flip syndrome task.} The four classes are no phase error and single-qubit $Z_1$, $Z_2$, or $Z_3$ errors on a GHZ-like state. The informative QL context contains $X_1X_2$ and $X_2X_3$ projectors. A diagonal $Z$ context is deliberately insufficient.

\textbf{Random stabilizer tasks.} We generate randomized $n$-qubit, $k$-generator stabilizer-style tasks. A full-rank binary parity matrix defines $k$ commuting Pauli generators after local basis maps $Z\mapsto X,Y,Z$. Each syndrome pattern gives one class, so $k=4$ yields 16 classes with chance accuracy $1/16=0.0625$. We compare the true QL context, a diagonal $Z$ context, and a wrong non-diagonal QL context.

\textbf{Mixed literal-pool task.} We form a larger candidate pool containing true context literals, wrong-context literals, diagonal $Z$ literals, and random Pauli distractors. This tests whether success depends on handing the learner only the correct context.

\textbf{Context-budget ablation.}
To test whether QL-TM succeeds only because a complete true context is supplied, we also consider partial-context tasks. For a $k$-generator stabilizer problem, only $b\leq k$ true generators are made available. If only $b$ true generators are present, the available projector features can distinguish at most $2^b$ syndrome groups out of $2^k$ classes. In a balanced task, the expected accuracy therefore follows the ladder $2^{b-k}$. This ablation checks whether performance degrades predictably when physically necessary quantum propositions are removed, rather than remaining perfect merely because the feature space is enlarged.

\subsection{Noise and Sampling}

For a pure state $\rho$, depolarizing noise is simulated as
\begin{equation}
    \rho_p = (1-p)\rho + p I/d,
\end{equation}
where $d=2^n$. We also simulate finite-shot estimates by binomial sampling of Born probabilities, readout noise by flipping binary outcomes with probability $r$, and coherent local perturbations by applying small random single-qubit rotations before measurement.

\subsection{Baselines and Metrics}

We report four models: ProjectorClauseMiner, TinyProjectorTsetlin, a nearest-prototype baseline, and a dependency-free ridge one-vs-rest classifier. The baselines separate projector-feature expressivity from a specific clause learner. We report accuracy, standard deviation across seeds/tasks, average literals per clause, true-generator coverage, and distractor literals selected.

\paragraph{Statistical reporting.}
Means and standard deviations are reported over the seeds/tasks in Table~\ref{tab:protocol}. We do not claim significance for small baseline differences; the principal claims concern large context separations, true-generator recovery, and the predicted $2^{b-k}$ context-budget trend. Formal paired significance testing is left to larger future studies.

Table~\ref{tab:notation} summarizes the notation used in the experiments and analysis. Exact/noiseless rows use Born probabilities directly. Finite-shot rows replace each probability by an empirical frequency from repeated Bernoulli trials, matching the way a commuting Pauli projector would be estimated from measurement counts.

\begin{table}[t]
\caption{Notation and experimental quantities.}
\label{tab:notation}
\centering
\begin{tabular}{ll}
\toprule
Symbol/quantity & Meaning \\
\midrule
$P_g^\pm$ & Pauli projector $(I\pm g)/2$ \\
$C_j$ & Commuting projector clause \\
$c_j(\rho)$ & Clause activation $\Tr(\rho C_j)$ \\
$k$ & Number of stabilizer generators/classes $2^k$ \\
True coverage & Fraction of true generators selected \\
Selected lit. & Average literals per learned class clause \\
Wrong QL & Non-diagonal projectors in the wrong basis \\
\bottomrule
\end{tabular}
\end{table}

\begin{table*}[t]
\caption{Experimental protocol and hyperparameters. Reported standard deviations are computed over the seeds and random tasks used by the corresponding experiment.}
\label{tab:protocol}
\centering
\resizebox{\textwidth}{!}{%
\begin{tabular}{ll}
\toprule
Quantity & Value \\
\midrule
Train/test split & Random stratified 70/30 train/test split using the package routine \texttt{train\_test\_split}; fixed seed per run \\
Bell/phase-flip context experiment & 80 samples per class, 10 seeds for feature-baseline results \\
Random stabilizer experiment & $n\in\{5,6\}$, $k=4$ generators, 16 classes, 6 random tasks, 3 seeds, 50 samples per class \\
Mixed literal-pool experiment & $n\in\{5,6\}$, $k=4$, 6 random tasks, 3 seeds, 60 samples per class \\
Context-budget ablation & $n\in\{5,6\}$, $k=4$, 6 random tasks, 3 seeds, 60 samples per class; all subsets for each true-generator budget $b\in\{0,1,2,3,4\}$ \\ 
Finite-shot/noise grid & Depolarizing $p\in\{0,0.05,0.1,0.2,0.4,0.6\}$; shots $\in\{\mathrm{exact},16,32,64,128,256\}$ \\
Additional noise & Coherent rotation scale $\in\{0,0.1,0.2,0.3\}$; readout noise $\in\{0,0.02,0.05\}$ \\
ProjectorClauseMiner & Minimum margin $0.03$; one positive clause per class in reported stabilizer experiments \\
TinyProjectorTsetlin & Automaton states parameter $N=16$; threshold $\tau=0.55$; 8 epochs for context/noise runs and 10 epochs for random/mixed stabilizer runs \\
Mixed-pool composition & True-context literals, wrong-context literals, diagonal $Z$ literals, and 32 random Pauli-observable distractors, each with $\pm$ projectors \\
Wrong QL context & Non-diagonal Pauli-projector context generated in an incompatible/random basis, excluding the true stabilizer generators \\
Implementation & NumPy, pandas, matplotlib; no dependency on TMU or pyTsetlinMachine \\
\bottomrule
\end{tabular}}
\end{table*}

\paragraph{Hyperparameter policy.}
Hyperparameters in Table~\ref{tab:protocol} were fixed within experiment families and were not tuned against the reported test sets. We use $N=16$, a common $\tau=0.55$, and short fixed epoch budgets because the goal is controlled clause recovery rather than optimized predictive performance. A systematic sensitivity analysis over $N$, $\tau$, clause multiplicity, and feedback schedules remains future work.

Consequently, the reported results should be interpreted as evidence for the
proposed clause semantics under a fixed prototype configuration, not as
evidence that these hyperparameters are optimal. A systematic sensitivity
analysis over $N$, $\tau$, clause multiplicity, and feedback schedules remains
future work.

\section{Results}

\subsection{Contextual Expressivity}

Table~\ref{tab:context} and Fig.~\ref{fig:context} show the basic separation. In the exact/noiseless setting, QL projector contexts attain 1.000 mean accuracy on the Bell and phase-flip tasks across all four evaluated models. Diagonal $Z$-only contexts collapse: Bell accuracy is about 0.50 because the diagonal context cannot distinguish the $+$ and $-$ Bell phases, and phase-flip accuracy is about 0.20--0.24 because the relevant syndrome is in an $X$ context.

\begin{figure}[t]
    \centering
    \includegraphics[width=\columnwidth]{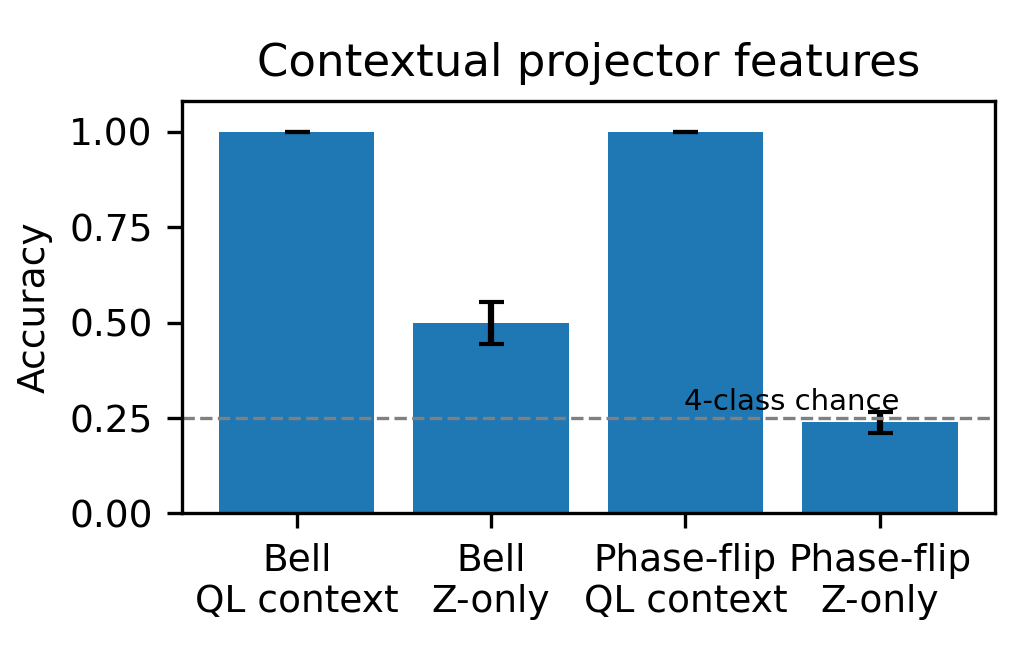}
    \caption{Context advantage for TinyProjectorTsetlin. Non-diagonal QL contexts solve the tasks, while diagonal $Z$-only contexts lose the relevant phase/syndrome information.}
    \label{fig:context}
\end{figure}

\begin{table}[t]
\caption{Context advantage across models. Values are mean $\pm$ standard deviation.}
\label{tab:context}
\centering
\resizebox{\columnwidth}{!}{%
\begin{tabular}{lcccc}
\toprule
Model & Bell QL & Bell $Z$ & Phase QL & Phase $Z$ \\
\midrule
ProjectorClauseMiner & 1.000$\pm$0.000 & 0.500$\pm$0.055 & 1.000$\pm$0.000 & 0.239$\pm$0.028 \\
TinyProjectorTsetlin & 1.000$\pm$0.000 & 0.500$\pm$0.055 & 1.000$\pm$0.000 & 0.239$\pm$0.028 \\
RidgeOneVsRest & 1.000$\pm$0.000 & 0.449$\pm$0.032 & 1.000$\pm$0.000 & 0.201$\pm$0.020 \\
NearestPrototype & 1.000$\pm$0.000 & 0.500$\pm$0.055 & 1.000$\pm$0.000 & 0.239$\pm$0.028 \\
\bottomrule
\end{tabular}}
\end{table}

The learned clauses are also physically meaningful. For the Bell task, both clause learners recover
\begin{align*}
|\Phi^+\rangle &: P_{ZZ}^{+}\wedge P_{XX}^{+}, &
|\Phi^-\rangle &: P_{ZZ}^{+}\wedge P_{XX}^{-},\\
|\Psi^+\rangle &: P_{ZZ}^{-}\wedge P_{XX}^{+}, &
|\Psi^-\rangle &: P_{ZZ}^{-}\wedge P_{XX}^{-}.
\end{align*}
For phase-flip syndromes, the recovered clauses are
\begin{align*}
\mathrm{None} &: P_{X_1X_2}^{+}\wedge P_{X_2X_3}^{+}, &
Z_1 &: P_{X_1X_2}^{-}\wedge P_{X_2X_3}^{+},\\
Z_2 &: P_{X_1X_2}^{-}\wedge P_{X_2X_3}^{-}, &
Z_3 &: P_{X_1X_2}^{+}\wedge P_{X_2X_3}^{-}.
\end{align*}
The corresponding diagonal clauses either merge Bell phase pairs or degenerate to \texttt{TRUE} for all phase-flip classes.

\subsection{Randomized Stabilizer Generalization}

Fig.~\ref{fig:random} and Table~\ref{tab:random} show the stronger randomized task. For $k=4$, the classifier must distinguish 16 syndrome classes. The true QL context achieves 1.000 accuracy for $n=5$ and $n=6$, while wrong QL contexts remain near chance. Diagonal $Z$ contexts sometimes extract weak accidental information but remain far below the true context.

\begin{figure}[t]
    \centering
    \includegraphics[width=\columnwidth]{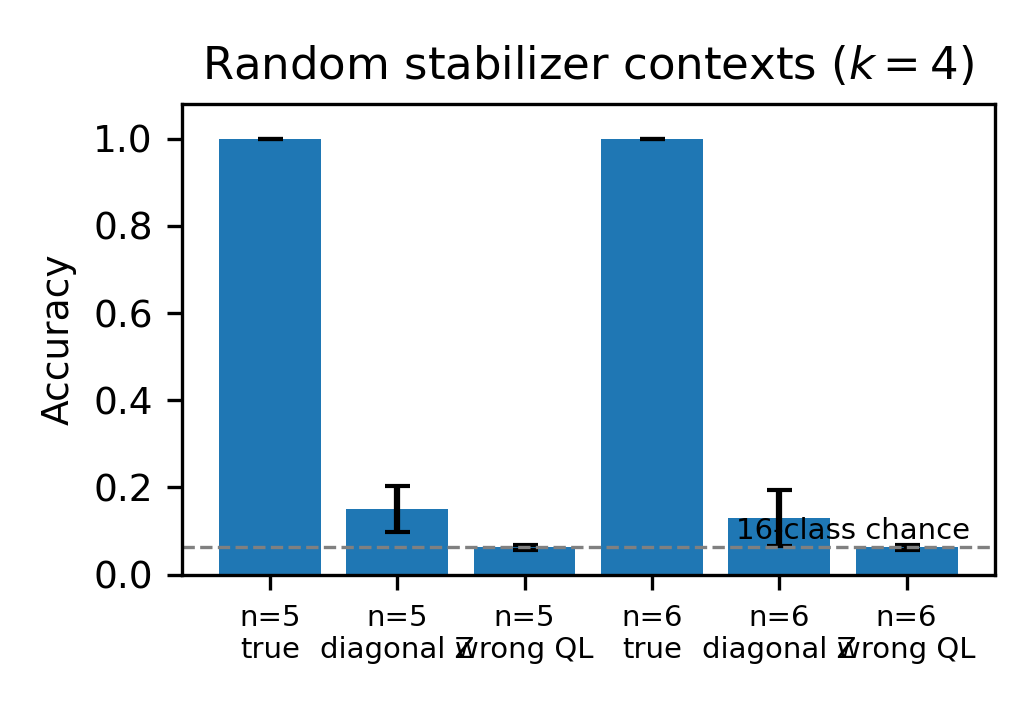}
    \caption{Random 16-class stabilizer tasks using TinyProjectorTsetlin. The correct QL context solves the task; wrong non-diagonal contexts remain at the chance level.}
    \label{fig:random}
\end{figure}

\begin{table}[t]
\caption{Random stabilizer results for TinyProjectorTsetlin, exact/noiseless setting, $k=4$ generators.}
\label{tab:random}
\centering
\begin{tabular}{lccc}
\toprule
$n$ & Feature set & Accuracy & Lit./clause \\
\midrule
5 & True QL & 1.000$\pm$0.000 & 4.00 \\
5 & Diagonal $Z$ & 0.150$\pm$0.052 & 2.50 \\
5 & Wrong QL & 0.062$\pm$0.006 & 0.00 \\
6 & True QL & 1.000$\pm$0.000 & 4.00 \\
6 & Diagonal $Z$ & 0.130$\pm$0.064 & 1.17 \\
6 & Wrong QL & 0.062$\pm$0.006 & 0.00 \\
\bottomrule
\end{tabular}
\end{table}

The representative random task clauses have the form
\begin{equation}
    s_{\sigma_1\sigma_2\sigma_3\sigma_4}:\quad
    P_{g_1}^{\sigma_1}\wedge P_{g_2}^{\sigma_2}\wedge P_{g_3}^{\sigma_3}\wedge P_{g_4}^{\sigma_4},
\end{equation}
which is exactly the expected stabilizer-syndrome proposition. This supports the claim that QL-TM is learning quantum-context clauses, not memorizing Bell-state labels.

\subsection{Mixed Literal Pools and Pruning}

The mixed-pool experiment asks whether the learner can handle a candidate library that contains true, wrong, diagonal, and random projector literals. Table~\ref{tab:mixed} shows that the mixed pool retains perfect accuracy for both $n=5$ and $n=6$ while achieving full true-generator coverage. When the true literals are removed, performance becomes unstable and substantially lower. Wrong contexts remain near chance.

\begin{table}[t]
\caption{Mixed-pool results for TinyProjectorTsetlin, exact/noiseless setting, $k=4$. Coverage is the fraction of true generators selected.}
\label{tab:mixed}
\centering
\resizebox{\columnwidth}{!}{%
\begin{tabular}{llccc}
\toprule
$n$ & Feature pool & Accuracy & Selected lit. & Coverage \\
\midrule
5 & True only & 1.000$\pm$0.000 & 4.00 & 1.00 \\
5 & Mixed pool & 1.000$\pm$0.000 & 7.00 & 1.00 \\
5 & Mixed w/o true & 0.419$\pm$0.427 & 3.00 & 0.00 \\
5 & Diagonal $Z$ & 0.186$\pm$0.155 & 2.00 & 0.00 \\
5 & Wrong QL & 0.049$\pm$0.011 & 0.00 & 0.00 \\
6 & True only & 1.000$\pm$0.000 & 4.00 & 1.00 \\
6 & Mixed pool & 1.000$\pm$0.000 & 7.83 & 1.00 \\
6 & Mixed w/o true & 0.351$\pm$0.331 & 3.83 & 0.00 \\
6 & Diagonal $Z$ & 0.179$\pm$0.067 & 3.33 & 0.00 \\
6 & Wrong QL & 0.049$\pm$0.011 & 0.00 & 0.00 \\
\bottomrule
\end{tabular}}
\end{table}

The mixed-pool learner sometimes retains redundant distractors in addition to the four true literals. Fig.~\ref{fig:prune} reports a small post-hoc pruning ablation: when redundant non-true literals are removed while retaining the true generators, clauses compress back to the $k=4$ stabilizer form with unchanged exact/noiseless accuracy. This suggests that explicit sparsification should be part of a future full-scale QL-TM implementation.

\begin{figure}[t]
    \centering
    \includegraphics[width=\columnwidth]{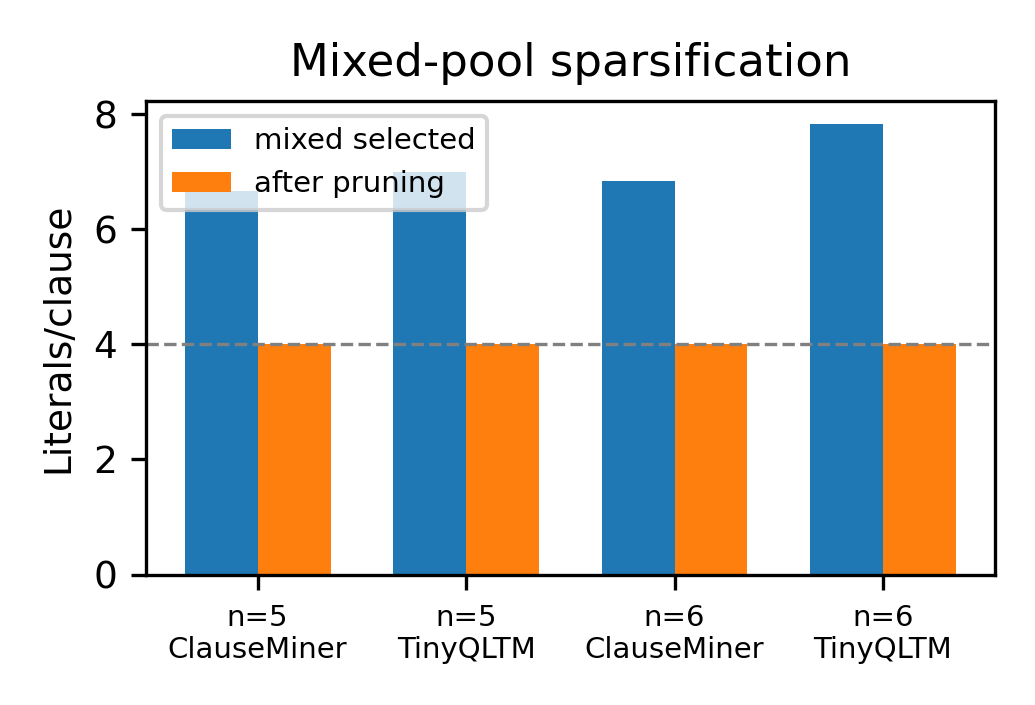}
    \caption{Post-hoc mixed-pool pruning. Mixed-pool clauses include all true stabilizer generators plus redundant distractors; pruning compresses clauses back to four literals while preserving accuracy in the exact/noiseless task.}
    \label{fig:prune}
\end{figure}

The next ablation asks the complementary question: if some true generators are removed entirely rather than merely hidden among distractors, does performance degrade according to the remaining syndrome information?

\subsection{Context-Budget Ablation}

The mixed-pool experiment shows that QL-TM can recover the true generators when they are present among distractors. A stricter question is whether performance degrades when some of the true quantum propositions are absent. Table~\ref{tab:budget} reports the context-budget ablation for $k=4$, where only $b$ of the four true stabilizer generators are available. The expected separability ladder is $2^{b-4}$: with $b=0$ the task is at the 16-class chance level, while $b=4$ restores full separability.

\begin{table}[t]
\caption{Context-budget ablation for TinyProjectorTsetlin, exact/noiseless setting, $k=4$. Values are mean $\pm$ standard deviation over 6 random tasks, 3 seeds, and all true-generator subsets at each budget.}
\label{tab:budget}
\centering
\resizebox{\columnwidth}{!}{%
\begin{tabular}{cccc}
\toprule
Available true generators $b/4$ & Expected & $n=5$ acc. & $n=6$ acc. \\
\midrule
0/4 & 0.0625 & 0.049$\pm$0.011 & 0.049$\pm$0.011 \\
1/4 & 0.1250 & 0.113$\pm$0.013 & 0.113$\pm$0.013 \\
2/4 & 0.2500 & 0.240$\pm$0.023 & 0.240$\pm$0.023 \\
3/4 & 0.5000 & 0.495$\pm$0.027 & 0.495$\pm$0.027 \\
4/4 & 1.0000 & 1.000$\pm$0.000 & 1.000$\pm$0.000 \\
\bottomrule
\end{tabular}}
\end{table}

The observed accuracies closely track the predicted $2^{b-k}$ separability ladder: near chance with no true generator, approximately $0.125$, $0.25$, and $0.5$ with one, two, and three generators, and 1.000 with all four. Because this hierarchy is determined by the available stabilizer information before learning, it provides controlled evidence that performance is governed by physically meaningful projector propositions rather than raw feature dimensionality.

\subsection{Finite-Shot and Noise Robustness}

The robustness grid varies depolarizing noise, finite shots, coherent local rotations, and readout noise. The most severe plotted slice uses depolarizing noise $p=0.6$, rotation scale 0.2, and readout noise 0.02. In this setting, TinyProjectorTsetlin still obtains 0.973 accuracy for Bell states at only 16 shots and 0.944 for phase-flip syndromes; by 64 shots both tasks recover 1.000 accuracy. Fig.~\ref{fig:noise} shows this behavior. The reason exact depolarizing noise is mild for these tasks is that it shrinks confidence without changing the stabilizer syndrome ordering; degradation appears when finite-shot and readout uncertainty make the estimated projector probabilities ambiguous.

\begin{figure}[t]
    \centering
    \includegraphics[width=\columnwidth]{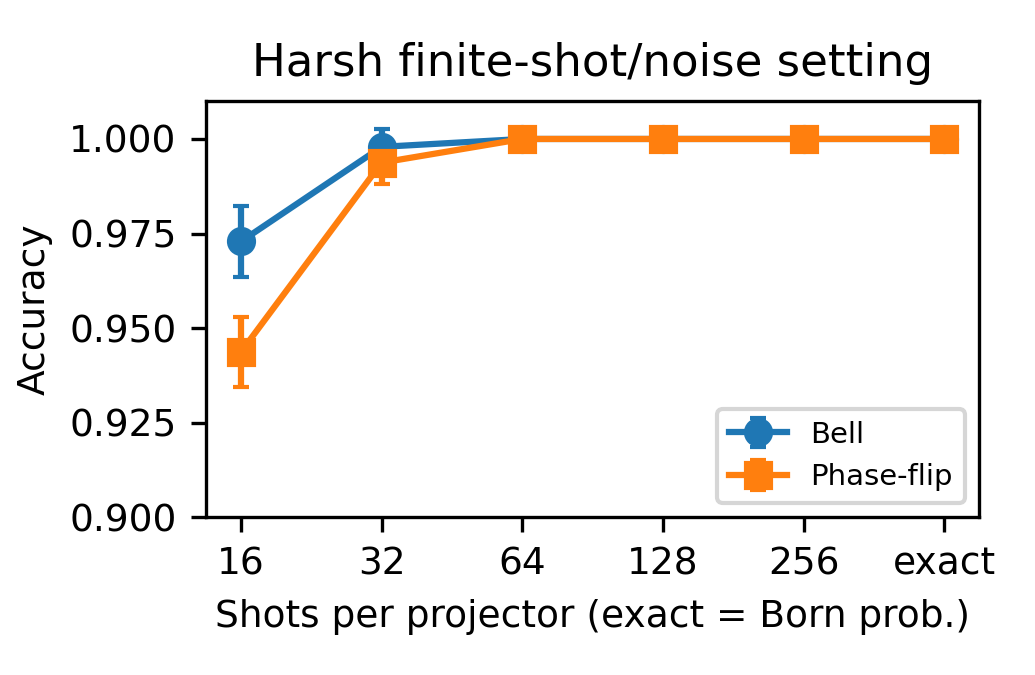}
    \caption{Finite-shot robustness under a harsh setting: depolarizing noise $p=0.6$, coherent rotation scale 0.2, and readout noise 0.02. Accuracy degrades mainly at very low shot counts.}
    \label{fig:noise}
\end{figure}

\section{Discussion}

Three scoped conclusions follow. First, the diagonal-reduction theorem identifies ordinary Boolean TM semantics as a special case, while the Bell, phase-flip, and randomized stabilizer experiments show that non-diagonal commuting contexts expose phase/syndrome information absent from diagonal features. Second, learned clauses have direct joint-eigenspace interpretations; mixed-pool experiments recover all true generators among distractors, although redundant literals motivate native sparsification. Third, the context-budget results follow the externally predicted $2^{b-k}$ hierarchy as true generators are removed, providing the strongest evidence that performance tracks quantum-context information rather than feature count.

The scope remains deliberately narrow. The tasks are Bell/GHZ/stabilizer based because their correct projector propositions are analytically known; they do not establish superiority over VQCs, quantum kernels, or other end-to-end QML models. Likewise, the product-of-marginals quantity is only a training surrogate; exact joint-context learning is not yet validated beyond the deterministic stabilizer setting.

\subsection{Measurement and Limitations}

Because each clause lies in one commuting context, its literals can in principle be estimated from compatible measurements, suggesting a future measurement-cost objective based on contexts and shots. The present study, however, is entirely simulated and does not model hardware calibration drift or device-specific measurement constraints. TinyProjectorTsetlin is a minimal prototype rather than a full score-gated production TM, no task-specific hyperparameter sweep or formal significance-testing campaign is reported, and less structured state families remain unevaluated. Future work should therefore study direct joint-projector feedback, context selection, multi-clause learning, native sparsity, hardware-derived data, and comparisons on broader QML tasks.

\section{Conclusion}

We introduced \QLTM, a hybrid quantum-state classifier whose Tsetlin-style literals are projectors and whose clauses are restricted to commuting contexts. The framework reduces exactly to Boolean TM clause semantics in diagonal contexts and yields stabilizer/syndrome semantics for Pauli projectors. Controlled experiments demonstrate contextual expressivity, randomized stabilizer generalization, true-literal recovery from mixed pools, the predicted context-budget degradation, and finite-shot robustness. These results establish a foundation for quantum-logic Tsetlin learning without claiming quantum advantage or state-of-the-art end-to-end QML performance.

\section*{Reproducibility Note}

Source code, raw results, figures, and learned clauses are available at \newline
\url{https://github.com/FraQTech/quantum-logic-tsetlin-machines}. Table~\ref{tab:protocol} records the data-generation settings, seeds, sample counts, and fixed hyperparameters used in the reported experiments.

\bibliographystyle{IEEEtran}
\bibliography{refs}

\end{document}